\documentclass[preprintnumbers,aps,prd,showpacs,floatfix,superscriptaddress,nofootinbib,letterpaper]{revtex4-2}
\usepackage[utf8]{inputenc}
\usepackage{graphicx} 
\usepackage{xcolor}
\usepackage{bm}
\usepackage{amsmath, amsthm, amssymb}
\newtheorem{theorem}{Theorem}
\usepackage{hyperref}
\usepackage[top = 2cm, bottom = 2cm, right = 2cm, left =2cm]{geometry}
\begin{document}

\title{On regular black string spacetimes in nonlinear electrodynamics}

\author{G. Alencar\footnote{E-mail: geova@fisica.ufc.br}}
\affiliation{Departamento de F\'isica, Universidade Federal do Cear\'a, Caixa Postal 6030, Campus do Pici, 60455-760 Fortaleza, Cear\'a, Brazil}

\author{V. H. U. Borralho\footnote{Author to whom any correspondence should be addressed.}\footnote{E-mail: victorborralho@fisica.ufc.br}}
\affiliation{Departamento de F\'isica, Universidade Federal do Cear\'a, Caixa Postal 6030, Campus do Pici, 60455-760 Fortaleza, Cear\'a, Brazil}

\author{T. M. Crispim\footnote{E-mail: tiago.crispim@fisica.ufc.br}}
\affiliation{Departamento de F\'isica, Universidade Federal do Cear\'a, Caixa Postal 6030, Campus do Pici, 60455-760 Fortaleza, Cear\'a, Brazil}

\author{M. S. Cunha\footnote{E-mail: marcony.cunha@uece.br}}
\affiliation{Centro de Ciências e Tecnologia, Universidade Estadual do Ceará, 60714-903, Fortaleza, Ceará, Brazil}


\begin{abstract}
In this work, we investigate the coupling of General Relativity with Nonlinear Electrodynamics (NED), governed by a general Lagrangian $\mathcal{L}(\mathcal{F})$, to address the axial singularity of four-dimensional black strings. Through a model-independent analysis, we scrutinize the viability of regular configurations by extending no-go theorems, originally formulated for spherical spacetimes, to cylindrical symmetries. We provide a comprehensive mathematical proof that regular, purely electric black strings cannot be generated by any NED Lagrangian that recovers the Maxwell limit in the weak-field regime, establishing a fundamental constraint for cylindrical topologies. Despite these limitations, we employ specific mathematical frameworks to construct new exact solutions for black strings, including cylindrical analogues of the well-known Bardeen and Hayward regular black hole classes. Each solution is analytically derived, and we demonstrate that their curvature invariants remain finite everywhere, effectively replacing the axial singularity with a regular core. Furthermore, we evaluate the physical consistency of these new metrics by subjecting them to stringent causality and unitarity constraints. Our results provide a comprehensive classification of the conditions under which NED can regularize cylindrical spacetimes and offer new insights into how topological differences between spherical and axial symmetries influence the global structure and the physical viability of non-singular gravitational objects in nonlinear gauge theories.
\end{abstract}

\maketitle

\section{Introduction}
General Relativity (GR) provides an exceptionally accurate description of the macroscopic universe, yet it remains intrinsically plagued by the emergence of spacetime singularities \cite{Penrose:1965, Hawking:1970, Hawking:1966}. At these pathological regions, curvature invariants diverge, and the predictive power of the classical theory completely breaks down. The resolution of this physical limitation remains a major open question, often addressed phenomenologically through the construction of regular (singularity-free) spacetimes. Concurrently, the pursuit of fundamental theories at high energies, such as string theory and its extended objects like D-branes \cite{Horowitz:1991, Polchinski:1995, Duff:1995}, has strongly motivated the exploration of non-spherical horizon topologies \cite{Emparan:2002}. Furthermore, the holographic principle heavily relies on extended black solutions in Anti-de Sitter backgrounds \cite{Maldacena:1998}. Although extended objects are typically associated with extra dimensions \cite{Chamblin:1999by, Chamblin:2000ra, Kanti:2002fx, Hirayama:2001bi, Crispim:2024yjz, Estrada:2024lhk, Kanti:2001cj, Alencar:2026dsq}, string-like geometries with cylindrical horizons, commonly known as black strings, can be remarkably realized within four-dimensional spacetime, following the pioneering work of Lemos \cite{LEMOS}. 
    
However, these classical 4D black strings inherently harbor a curvature singularity along their symmetry axis, demanding theoretical mechanisms for their regularization \cite{Lima:2023jtl, Lima:2023arg, Lima:2022pvc, Bronnikov:2023aya}. Furthermore, black strings are not necessarily isolated systems, \textit{i.e}., they can be surrounded by various types of matter and fields \cite{Lemos_1996, Alencar:2026xfr, Cunha:2022kep, Hendi:2010kv, Ali:2019mxs, Bronnikov:2023aya}, much like their spherical counterparts \cite{Cardoso:2021wlq, Konoplya:2021ube, Medved:2003rga, Cunha:2015yba, Herdeiro:2015waa, Herdeiro:2014goa, Chan:2025oux, Rahmatov:2025gpk, Bakhodirov:2025tpw}. In this context, we can refer to the so-called ``dirty'' black strings, where the standard background geometry is modified by the presence of this surrounding environment. Consequently, their gravitational features can be deformed by these fields, potentially altering the geometry. In this case, the gravitational backreaction of these surrounding fields deforms the standard spacetime geometry, and strikingly, this matter-induced modification can be profound enough to resolve the inherent axial singularity.

 Cylindrically symmetric spacetimes have a long and rich history in General Relativity. Besides black strings, this class includes the Levi-Civita vacuum solution \cite{LeviCivita:1919} Weyl-type configurations, cosmic strings, Melvin magnetic universes, and a variety of self-gravitating electromagnetic systems \cite{Bronnikov:2019clf, Melvin:1963qx,Weyl:1917}. Such geometries have played an important role in the investigation of exact solutions, topological defects, gravitational collapse, and electromagnetic self-confinement mechanisms. More recently, cylindrically symmetric configurations supported by nonlinear electrodynamics have also attracted attention, leading to the construction of self-sustaining solutions and the establishment of several regularity and no-go results \cite{Bronnikov:2003ru,Sokolov:2022pxh}.

In this context, a particularly compelling type of surrounding environment is the electromagnetic field, which was first incorporated into the black string framework by Lemos \cite{Lemos_1996}. However, the inclusion of electromagnetic sources governed by Nonlinear Electrodynamics (NED) within cylindrical spacetimes remains largely underexplored, with only a few instances in the literature (see, e.g., \cite{Bronnikov:2023aya,Alencar:2024yvh,Lima:2023arg}). This class of nonlinear theories was originally introduced by Born and Infeld in 1934 to resolve the divergence of the electron's self-energy in classical electrodynamics \cite{Born:1934gh}. More general extensions of the Einstein--Maxwell framework have also been proposed, including nonlinear constitutive electromagnetic tensors and non-Riemannian geometric contributions \cite{Cotton2021}. Decades later, this framework was naturally extended to GR, where NED models have proven remarkably successful in generating regular black hole solutions \cite{Ayon-Beato:1998hmi, Hayward:2005gi, Fan:2016hvf, Bronnikov:2024izh, Bolokhov:2024sdy, Rodrigues:2018bdc, Rodrigues:2017yry, deSSilva:2024fmp, Rodrigues:2022qdp}, and other non-singular configurations \cite{Alencar:2025jvl, Rodrigues:2023vtm, Alencar:2024yvh, Rodrigues:2022mdm, Pereira:2024rtv, Furtado:2022tnb, Bronnikov:2022bud, Crispim:2024lzf, Crispim:2024dgd, Silva:2025eip}, effectively smoothing out curvature singularities in spherically symmetric spacetimes. Recent studies have shown that, in regular black holes coupled to nonlinear electrodynamics, the Komar mass and the electric and magnetic charges are not independent, while new theoretical constraints have also been established for regular electrically charged and dyonic solutions \cite{Bokulic:2025brf}.

Despite its theoretical appeal and success in singularity resolution, the application of NED involves significant physical and mathematical challenges. From a fundamental perspective, regularizing a spacetime typically requires the nonlinear matter source to violate certain standard energy conditions -- most notably the Strong Energy Condition (SEC) -- in the deep core region to evade the singularity theorems. Furthermore, the highly nonlinear nature of the associated field equations makes it notoriously difficult to obtain exact, analytical solutions, particularly for extended non-spherical topologies such as black strings. Beyond energy conditions, the nonlinear nature of the electromagnetic field can introduce severe causality issues. In NED frameworks, photons propagate along the null geodesics of an effective geometry rather than the background spacetime. This discrepancy often leads to pathological behaviors, such as superluminal signal propagation or birefringence \cite{Novello:1999pg, Russo:2022qvz, deMelo:2014isa, Russo:2024llm, Abalos:2015gha, dePaula:2024yzy}. Consequently, constructing a physically sound NED Lagrangian that successfully regularizes the axial singularity of a black string while strictly preserving causality and unitarity remains a highly non-trivial mathematical task.

Motivated by these theoretical challenges, we systematically investigate the coupling of GR with NED in the context of cylindrical symmetry. In its most general formulation, a NED theory is governed by a Lagrangian density $\mathcal{L}(\mathcal{F}, \mathcal{H})$ that depends on both fundamental electromagnetic invariants: $\mathcal{F} = F_{\mu\nu}F^{\mu\nu}$ and the pseudo-invariant $\mathcal{H} = \tilde{F}_{\mu\nu}F^{\mu\nu}$, where $ \tilde{F}_{\mu\nu}=\star F_{\mu\nu}$ denotes the Hodge dual of the Faraday tensor. However, to establish a rigorous foundational analysis, this paper restricts itself to the simpler, yet remarkably rich, class of models defined purely by $\mathcal{L} = \mathcal{L}(\mathcal{F})$.

It is crucial to emphasize that our primary focus is not merely to obtain new exact solutions, but rather to perform a comprehensive, model-independent analysis of the fundamental field equations. When dealing with the NED-GR system, one must carefully distinguish between purely electric, magnetic, and dyonic configurations. In spherically symmetric scenarios, this general analysis leads to well-known no-go theorems~\cite{Bronnikov2023,Bronnikov:2000vy}, demonstrating that purely electric configurations with a standard Maxwellian weak-field limit inevitably fail to produce a regular center, even if the field's energy is finite. An important contribution to the study of cylindrically symmetric configurations in nonlinear electrodynamics was provided by Bronnikov, Shikin and Sibileva \cite{Bronnikov:2003ru}, who established several no-go theorems for regular electric and magnetic fields and discussed the existence of horizonless solitonic solutions. In contrast, the present work is devoted to Lemos black-string spacetimes, which possess Anti-de Sitter asymptotics and may contain event horizons. Within this framework, we investigate regularity conditions for electrically, magnetically, and dyonically charged configurations, and construct novel analytical black-string solutions, analyzing their physical consistency as well as their causal and dynamical viability.

The remainder of this paper is organized as follows. In Sec.~\ref{sec:main_equations}, we establish the theoretical framework, presenting the main equations and explicitly demonstrating the no-go theorems for cylindrical spacetimes. This section includes a detailed analysis of the regularity and asymptotic conditions for magnetic, electric, and dyonic configurations. In Sec.~\ref{sec:charged_bs}, as an application of our analysis, we construct specific charged black string solutions, focusing on Euler-Heisenberg electrodynamics (Case I) and Logarithmic electrodynamics (Case II). In Sec.~\ref{sec:regular_bs}, we investigate the formation of entirely regular charged black strings, directing our attention to the Bardeen class (Case I) and Hayward class (Case II).  Section~\ref{sec:causality} is devoted to the analysis of causality and unitarity constraints to ensure the physical validity of the nonlinear models. Finally, our conclusions and final remarks are presented in Sec.~\ref{sec:conclusions}.

We adopt the metric signature $(-,+,+,+)$ and use the geometric units where $G = c =\hbar= 1$.

\section{Main equations and no-go theorems}\label{sec:main_equations}
In order to obtain black string solutions supported by nonlinear electromagnetic fields, we consider the coupling of NED to gravity within the framework of the Einstein--Hilbert action with a cosmological constant. The total action is constructed by supplementing the gravitational sector with a nonlinear electromagnetic Lagrangian density $\mathcal{L}(\mathcal{F})$. 

As mentioned earlier, in its most general formulation, a NED theory is described by a Lagrangian density $\mathcal{L}(\mathcal{F},\mathcal{H})$, where $\mathcal{F}$ and $\mathcal{H}$ are the two independent electromagnetic invariants constructed from the field strength tensor $F_{\mu\nu}$. The scalar invariant is defined as $\mathcal{F}\equiv F_{\mu\nu}F^{\mu\nu}$, while the pseudo-scalar invariant $\mathcal{H}\equiv F_{\mu\nu}\tilde{F}^{\mu\nu}$ involves the dual electromagnetic tensor $\tilde{F}^{\mu\nu}$. This general dependence allows the theory to incorporate parity-violating contributions and more intricate nonlinear interactions between the electromagnetic field components \cite{Peskin1995,Plebanski1970}.

However, in order to simplify our analysis and focus on the essential nonlinear electromagnetic effects relevant for the solutions considered in this work, we restrict ourselves to the subclass of theories for which the Lagrangian depends solely on the invariant $\mathcal{F}$, namely
$\mathcal{L}=\mathcal{L}(\mathcal{F})$. This restriction is commonly adopted in the literature and already encompasses a wide variety of physically relevant NED models, including those capable of generating regular or modified black object geometries. The total action is then given by
\begin{equation}
    S=\int d^4x \sqrt{-g}\Big\lbrace R-2\Lambda + 2\mathcal{L}(\mathcal{F}) \Big\rbrace,
\end{equation}
where $R$ is the Ricci scalar, $\mathcal{L}(\mathcal{F})$ denotes the general Lagrangian density for NED. The field $F_{\mu\nu}$ is naturally described by the 2-form $\bm{F} = \frac{1}{2} F_{\mu\nu} dx^\mu \wedge dx^\nu$, which is given by the exterior derivative of the gauge connection 1-form $\bm{A} = A_\mu dx^\mu$, such that $\bm{F} = \bm{d}\bm{A}$. In components, this corresponds to $F_{\mu\nu}= 2\partial_{[\mu}A_{\nu]}$.

By varying the action with respect to the metric tensor $g^{\mu\nu}$, we obtain the Einstein field equations
\begin{equation}\label{efe}
G_{\mu\nu}+ \Lambda g_{\mu\nu}= T_{\mu\nu}.
\end{equation}
For a general nonlinear electromagnetic Lagrangian density $\mathcal{L}(\mathcal{F})$, the corresponding energy--momentum tensor is
\begin{equation}
  T_{\mu}^{\ \nu}=
\delta_{\mu}^{\ \nu}\,\frac{\mathcal{L}(\mathcal{F})}{2}
- 2\,\mathcal{L}_{\mathcal{F}}
\,F_{\mu}^{\ \lambda} F_{\lambda}^{\ \nu},
\end{equation}
where we defined $\mathcal{L}_{\mathcal{F}} \equiv \partial \mathcal{L} / \partial \mathcal{F}$.

To obtain the electromagnetic field equations, we vary the action with respect to the gauge potential $A^\mu$. This yields the generalized Maxwell equations associated with nonlinear electrodynamics,
\begin{equation}
\nabla_\mu \left( \mathcal{L}_{\mathcal{F}} F^{\mu\nu} \right) = 0.
\end{equation}
These equations, together with the Bianchi identity, can be compactly written using differential forms as
\begin{equation}
\bm{d}(\mathcal{L}_{\mathcal{F}} \star \bm{F}) = 0, \quad \bm{d}\bm{F} = 0,
\end{equation}
where $\star$ denotes the Hodge dual operator. In the linear limit $\mathcal{L}(\mathcal{F}) = -\mathcal{F}$, one recovers the standard Maxwell equations.

For the black string, we use the following ansatz for the static cylindrically symmetric metric \cite{LEMOS}
\begin{equation}
    ds^{2}=-f(r)dt^{2}+\frac{1}{f(r)}dr^{2}+r^{2}d\phi^{2}+\frac{r^{2}}{\ell^{2}}dz^{2},
\end{equation}
where \(f(r)\) can be expressed in terms of the linear mass density of the black string, \(m(r)\),
\begin{equation}\label{metric}
    f(r)=-\frac{4m(r)\ell}{r}+\frac{r^{2}}{\ell^{2}},
\end{equation}
assuming a negative cosmological constant of the form \(\Lambda=-3/\ell^{2}\).

Here \(\phi\) is a periodic coordinate, \(\phi \sim \phi + 2\pi\), while \(z \in \mathbb{R}\) is noncompact. Consequently, the spatial section of the horizon has cylindrical topology \(S^{1}\times \mathbb{R}\), which is the characteristic topology of a black string \cite{LEMOS, Lemos:2000un, Sheykhi:2020fqf}. This global identification distinguishes the present spacetime from the toroidal \(k=0\) topological black-hole case, for which both transverse directions are compact and the horizon topology is \(S^{1}\times S^{1}\) \cite{Lemos:2000un, Sheykhi:2020fqf}. Although the metric may be locally reminiscent of the \(k=0\) sector of topological black holes, the present solution should be interpreted globally as a cylindrically symmetric black string rather than as a toroidal black hole.\\
This distinction can be understood more systematically from the viewpoint of the isometry group acting on the two-dimensional surfaces defined by \(t=\mathrm{const}\) and \(r=\mathrm{const}\). As discussed in \cite{Sheykhi:2020fqf}, these surfaces admit a two-parameter Abelian isometry group \(G_2\), generated by two commuting Killing vectors. Different global identifications of the corresponding group orbits lead to different horizon topologies, even when the local metric structure remains unchanged. In particular, compactifying both Killing directions produces the flat torus \(T^2\simeq S^1\times S^1\), whereas compactifying only the angular direction while leaving the second Killing orbit noncompact yields the cylindrical topology \(R\times S^1\). Therefore, the distinction between toroidal and cylindrical horizons is not encoded in the local form of the metric itself, but rather in the global structure of the \(G_2\) orbits and the associated coordinate identifications.\\
In the present work, the Killing vector \(\partial_\phi\) generates closed orbits, while \(\partial_z\) generates an open noncompact orbit. Consequently, the horizon belongs to the cylindrical class \(R\times S^1\) discussed in \cite{LEMOS,Sheykhi:2020fqf}, rather than to the toroidal class \(S^1\times S^1\). From this perspective, the geometry is locally related to the \(k=0\) topological sector, but globally corresponds to a black-string spacetime.

For this metric, the components of the Einstein tensor are given by
\begin{align}
    G^{\;0} _{0}& = G^{\;1} _{1}=\frac{3}{\ell^2}-\frac{4   m '(r)\ell}{r^2} \\
    G^{\;2} _{2} &= G^{\;3} _{3}=\frac{3}{\ell^2} -\frac{2   m''(r) \ell}{r}
\end{align}
\subsection{Regularity and asymptotic conditions}
In this subsection, we analyze the conditions on the function $m(r)$ that guarantee the regularity of the spacetime. To this end, one may study the Ricci scalar $\mathcal{R} = g^{\mu\nu}R_{\mu\nu}$, the Kretschmann scalar $\mathcal{K}$, or the squared Ricci tensor $R_{\mu\nu}R^{\mu\nu}$. For this purpose, we compute the non-vanishing components of the Riemann tensor $R_{\alpha \beta}^{\;\;\;\;\rho\sigma}$
\begin{equation} \label{riemann}
     \begin{split}
        R_{01}^{\;\;01}&=\frac{4  m (r)\ell}{r^3}-\frac{4   m '(r)\ell}{r^2}+\frac{2   m ''(r)\ell}{r}-\frac{1}{\ell ^2} , \\
         R_{02}^{\;\;02}= R_{03}^{\;\;03}&= R_{12}^{\;\;12}= R_{13}^{\;\;13}=-\frac{2  m (r) \ell }{r^3}+\frac{2 m'(r)\ell }{r^2}-\frac{1}{\ell ^2} ,\\
    R_{23}^{\;\;23}&=\frac{4  m(r) \ell }{r^3}-\frac{1}{\ell ^2}.
    \end{split}
\end{equation}

Thus, the Kretschmann scalar $\mathcal{K}=R_{\alpha \beta \rho \sigma}R^{\alpha \beta \rho \sigma}$ can be written in terms of the squares of these components as
\begin{equation}
    \mathcal{K}=4( R_{01}^{\;\;01})^2+4( R_{02}^{\;\;02})^2+4( R_{03}^{\;\;03})^2+4(R_{12}^{\;\;12})^2+4(R_{13}^{\;\;13})^2+4( R_{23}^{\;\;23})^2.
\end{equation}
implying that $\mathcal{K}$ is finite if each component $R_{ab}^{\;\;ab}$ is finite.
Another important point is that, since the non-vanishing components of the Riemann tensor are written in the form $R_{ab}^{\;\;ab}$, the other curvature invariants can also be expressed in terms of these components and will likewise be finite provided that each of them is finite \cite{Bronnikov2023}. Therefore, the finiteness of these components constitutes a necessary and sufficient condition for the regularity of the spacetime \cite{Bronnikov:2012wsj}.

From Eq.~\eqref{riemann}, we observe that possible divergences may arise as $r \to 0$. 
Therefore, the function $m(r)$ must behave at least as
\begin{equation} \label{regularity1}
    m(r) \approx \mathcal{O}(r^3) \quad \text{as } r \to 0,
\end{equation}
which implies that
\begin{equation} \label{regularity2}
    m'(r)\approx \mathcal{O}(r^2) \quad \text{and} \quad m''(r)\approx \mathcal{O}(r).
\end{equation}
Under these conditions, the components of the Riemann tensor remain finite, 
which in turn guarantees the finiteness of all scalar curvature invariants. 
Consequently, the behavior of $m(r)$ near the origin completely determines 
the absence of scalar curvature singularities in the spacetime.

Since we are dealing with an asymptotically AdS solution $(\Lambda < 0)$, the mass function must satisfy
\begin{equation}\label{assintotico}
    m(r) \longrightarrow \mu \quad \text{as } r \to \infty,
\end{equation}
where $\mu$ is a constant. Consequently, $m'(r) \to 0$ and $m''(r) \to 0$ in this limit. Under these conditions, the AdS term dominates the asymptotic behavior of the curvature, ensuring that the spacetime approaches Anti-de Sitter geometry at infinity.

\subsection{Magnetic solutions}
Since we aim to generalize magnetically charged black string solutions, we consider a purely magnetic configuration compatible with cylindrical symmetry. In coordinates $(t,r,\phi,z)$, a convenient ansatz for the electromagnetic 2-form is
\begin{equation}\label{magnetico}\bm{F} = P  d\phi \wedge dz,
\end{equation}
where $P$ is a constant related to the magnetic charge per unit length of the string. This closed 2-form trivially satisfies the Bianchi identity, $\bm{dF} = 0$. This configuration corresponds to a radial magnetic field, since the magnetic flux pierces cylindrical surfaces of constant radius, whose normal vector points along the $r$ direction. With this choice, the electromagnetic invariant becomes
\begin{equation}
\mathcal{F}=\mathcal{F}(r) = \frac{2P^2\ell^2}{r^4}.
\end{equation}
In this case, it is possible to invert this relation and write $r = r(\mathcal{F})$, thereby simplifying the mathematical structure so that all equations depend solely on the invariant $\mathcal{F}$.

Because $\mathcal{F}$ depends only on $r$, the derivative $\mathcal{L}_{\mathcal{F}}$ is also a function of $r$ alone. Furthermore, due to the background symmetries, the Hodge dual of the Faraday tensor is a 2-form proportional to $dt \wedge dr$. Consequently, the generalized Maxwell equation is identically satisfied, $\bm{d}(\mathcal{L}_{\mathcal{F}} \star \bm{F}) = 0$, since the exterior derivative of a 2-form proportional to $dt \wedge dr$ with $r$-dependent coefficients vanishes identically (as $dr \wedge dt \wedge dr = 0$). This powerful feature holds regardless of the specific form of the nonlinear Lagrangian density.

 Substituting \eqref{metric} into \eqref{efe}, we obtain the following system of differential equations:
\begin{align} \label{mu}
&\frac{4m'(r)\ell}{r^2} + \frac{1}{2} \mathcal{L}(\mathcal{F}(r) ) = 0,  \\
&\frac{2 m''(r) \ell}{r}
+ \frac{1}{2} \left[\mathcal{L}(\mathcal{F}(r) )
- \frac{4 P^2 \ell^2 \mathcal{L_F}}{r^4}\right] = 0. \label{nu}
\end{align}
From these equations, we proceed to obtain solutions for different cases involving NED. We can obtain $m(r)$ by directly integrating Eq.~\eqref{mu}:
\begin{equation} \label{magneticoeqfinal}
    m'(r) = -\frac{r^2\mathcal{L}(\mathcal{F}(r))}{8\ell}
\end{equation}
To ensure the regularity of $m(r)$, we must require that as $r \to 0$ ($\mathcal{F} \to \infty$), $\mathcal{L}(\mathcal{F})$ tends to a finite value $\mathcal{L}_0$, and $\mathcal{L}_{\mathcal{F}} \to 0$. In this way, $m'(r) \approx \mathcal{O}(r^2)$, as shown in \eqref{regularity2}, ensuring a regular core instead of a curvature singularity. Analyzing the asymptotic behavior, in the case of an AdS solution, as $r \to \infty$ we have $\mathcal{F} \to 0$, $\mathcal{L}(\mathcal{F}) \approx -\mathcal{F}$, and $\mathcal{L}_{\mathcal{F}}$ approaches a constant value. In this regime, since $\mathcal{F} \propto r^{-4}$, it follows that
\begin{equation}
    m'(r) \propto r^{-2} \implies m(r) = \mu + \frac{C}{r}.
\end{equation}
Taking the limit $r \to \infty$, $m(r)$ approaches a constant value, as shown in \eqref{assintotico}.
\subsection{Electric solutions}
To generalize electrically charged solutions, we consider a purely electric configuration for the electromagnetic tensor. In this case, the 2-form can be written as
\begin{equation}\label{eletrico}
    \bm{F} = E(r)\, dt \wedge dr ,
\end{equation}
where $E(r)$ denotes the radial electric field, directed outward from the symmetry axis of the black string and perpendicular to its longitudinal extension along the $z$-direction. In other words, the electric field lines lie in the transverse $(r,\phi)$ plane and point radially away from (or toward) the axis, reflecting the cylindrical symmetry of the spacetime.

For this type of configuration, the Hodge dual of the Faraday tensor is given by
\begin{equation}
    \star \bm{F} = -\frac{r^2}{\ell} E(r)\, d\phi \wedge dz .
\end{equation}
Therefore, the generalized Maxwell equation becomes
\begin{equation}
    \bm{d}\!\left(\mathcal{L}_{\mathcal{F}} \star \bm{F}\right)
    =
    \partial_r\!\left(\frac{r^2}{\ell} \mathcal{L}_{\mathcal{F}} E(r) \right)
    dr \wedge d\phi \wedge dz
    =
    0 .
\end{equation}
Since $dr \wedge d\phi \wedge dz \neq 0$, it follows that
\begin{equation} \label{campoeletrico}    \partial_r\!\left(\frac{r^2}{\ell} \mathcal{L}_{\mathcal{F}} E(r) \right)=0
    \quad \Rightarrow \quad
    E(r)=\frac{q }{r^2 \mathcal{L}_{\mathcal{F}}} ,
\end{equation}
where $\mathcal{L}_{\mathcal{F}}=\frac{\partial \mathcal{L}}{\partial \mathcal{F}}$, and $q$ is an integration constant associated with the electric charge per unit length. In the linear limit, where $\mathcal{L}_{\mathcal{F}} = -1$, the standard Maxwell electric field is recovered. Therefore, in contrast to the purely magnetic configuration, the electric case does not satisfy the generalized Maxwell equation trivially. Instead, it leads to a nontrivial differential equation that must be solved to determine the electric field profile $E(r)$, whose explicit form depends on the specific choice of the nonlinear Lagrangian density.

With these choices, the electromagnetic invariant becomes
\begin{equation}\label{inveletrico}
    \mathcal{F}=\mathcal{F}(r)=-\frac{2 q^2 }{r^4 \mathcal{L_F}^2}.
\end{equation}
 Unlike the magnetic case, here it is not always possible to write $r = r(\mathcal{F})$, since this depends on the monotonicity of $\mathcal{F}(r)$, namely on $\frac{\partial \mathcal{F}}{\partial r}$.

From this point on, the resulting system of differential equations is
\begin{equation} \label{eletric1}
     \frac{4   m '(r)\ell}{r^2}+ \frac{\mathcal{L}(\mathcal{F}(r) )}{2}+\frac{2 q^2}{r^4\mathcal{L_F}} =0,
    \end{equation}
    \begin{equation}\label{eletric}
      \frac{2   m ''(r)\ell}{r} + \frac{\mathcal{L}(\mathcal{F}(r) )}{2}=0.
\end{equation}
In this way, we generalize the field equations in order to obtain electrically charged solutions.

For this configuration, the components of the energy-momentum tensor are given by
\begin{align}
    &T_{0}^{\;0}=T_{1}^{\;1}=\frac{1}{2} \mathcal{L}(\mathcal{F}(r))-\mathcal{F}(r) \mathcal{L_F} \\
    &T_{2}^{\;2}=T_{3}^{\;3}=\frac{1}{2} \mathcal{L}(\mathcal{F}(r))
\end{align}

To ensure a regular static cylindrically symmetric metric, each component of the energy-momentum tensor should be finite. Thus, we have the following theorem:
\begin{theorem} \label{teorema1}
Any purely electric static solution $(q \neq 0, P = 0)$ with cylindrical symmetry that is regular at the axis $r=0$ cannot be generated by a Lagrangian that recovers the Maxwell limit $\mathcal{L}(\mathcal{F}) \approx -\mathcal{F}$ for small values of $\mathcal{F}$ (weak field limit).
\end{theorem}

\begin{proof} 
For our metric, the Ricci tensor is diagonal, and the invariant $R_{\mu\nu} R^{\mu\nu}$ is given by the sum of the squares of $R_{\mu}^{\;\nu}$. As discussed previously, in order to ensure regularity, each term must remain finite at every point (including the origin) so that the metric itself is regular. This implies that, to obtain regular solutions, the components of the energy-momentum tensor must also be finite. 
Hence, we find that the term $\mathcal{F}{\mathcal{L}_{\mathcal{F}}}^2 \propto r^{-4}$ [Eq. \eqref{inveletrico}], which is singular at the origin. To guarantee that $|\mathcal{F}\mathcal{L}_{\mathcal{F}}|<\infty$, one must have  $\mathcal{F} \to 0$ and $\mathcal{L}_{\mathcal{F}} \to \infty$ as $r \to 0$; therefore, the Maxwell limit is not recovered at the axis $r=0$.
\end{proof}

The theorem \ref{teorema1} is already known in the case of spherically symmetric metrics ~\cite{Bronnikov2023,Bronnikov:2000vy}. Therefore, we show that its validity is not restricted to spherically symmetric spacetimes, but also holds in the cylindrical case.

Analyzing the asymptotic conditions for $m'(r)$, we can rewrite Eq.~\eqref{eletric} as
\begin{equation}
    m''(r) = -\frac{r \mathcal{L}(\mathcal{F})}{4 \ell}.
\end{equation}
Assuming $\mathcal{L}(\mathcal{F}) \approx -\mathcal{F}$ as $r \to \infty$, we have $\mathcal{F} \propto r^{-4}$. Therefore,
\begin{equation}
    m''(r) \propto r^{-3} \implies m(r) = \mu + \frac{C}{r}.
\end{equation}
In the limit $r \to \infty$, we recover a constant value for $m(r)$. With these considerations, for purely electric solutions we can satisfy only two out of the following three statements:
\begin{enumerate}
\item The spacetime is generated by a purely electric field with a non-zero charge.
\item The non-linear electrodynamics theory recovers the Maxwell limit in the weak-field regime at spatial infinity ($\mathcal{L}_{\mathcal{F}} \to \text{constant}$ as $r \to \infty$).
\item The geometry possesses a regular center, avoiding curvature singularities as $r \to 0$.
\end{enumerate}

\subsection{Dyonic solutions}
Now we analyze the situation where both the electric and magnetic sectors are present simultaneously. In this case, the ansatz for the 2-form reads
\begin{equation}
    \bm{F}=E(r) dt \wedge dr + P d\phi \wedge dz.
\end{equation}
For the dyonic case, the electromagnetic invariant is given by
\begin{equation}
    \mathcal{F}=\frac{2P^2}{r^4 \ell^2} - \frac{2q^2 }{r^4 \mathcal{L_F}^2}.
\end{equation}
From the generalized Maxwell equations, we find that the electric field has the form $E(r)={q }/{r^2 \mathcal{L}_{\mathcal{F}}}$; however, now $\mathcal{L}_{\mathcal{F}}=\mathcal{L}_{\mathcal{F}}(\mathcal{F}(E,P))$.
In this way, our system of differential equations reads
\begin{equation} \label{dionic1}
    \frac{\mathcal{L}(\mathcal{F}(r) )}{2}+ \frac{2 q^2}{r^4 \mathcal{L_F}}+\frac{4 m'(r)\ell}{r^2} =0
    \end{equation} 
    \begin{equation}\label{dionic2}
    \frac{\mathcal{L}(\mathcal{F}(r) )}{2}-\frac{2 P^2 \ell ^2 \mathcal{L_F}}{r^4}+\frac{2 m''(r)\ell }{r} =0.
\end{equation}
The inclusion of the magnetic sector does not modify the regularity analysis, since the leading-order behavior of the field invariants near the origin remains unchanged under the imposed conditions. 

Let the components of the energy-momentum tensor be:
\begin{equation}
   T_{0}^{\;0}=T_{1}^{\;1}=   \frac{\mathcal{L}(\mathcal{F}(r) )}{2}+ 2\mathcal{L_F} E(r)^2,
\end{equation}
and
\begin{equation}
   T_{2}^{\;2}=T_{3}^{\;3}=   \frac{\mathcal{L}(\mathcal{F}(r) )}{2}-\frac{2 P^2 \ell ^2 \mathcal{L_F}}{r^4}.
\end{equation}
From the previous considerations, each term must be finite at $r=0$ in order to obtain regular solutions. As $r\to0$, regularity requires $m'(r)=\mathcal O(r^2)$ and $m''(r)=\mathcal O(r)$, so that the last terms in Eqs.~\eqref{dionic1} and \eqref{dionic2} remain finite. The electric contribution in Eq.~\eqref{dionic1} behaves as $1/(r^4\mathcal L_F)$, implying that $\mathcal L_F$ must diverge at least as fast as $r^{-4}$. Conversely, the magnetic contribution in Eq.~\eqref{dionic2} behaves as $\mathcal L_F/r^4$, which requires $\mathcal L_F=\mathcal O(r^4)$. These two conditions are mutually incompatible. Therefore, no regular dyonic solution can satisfy both requirements simultaneously. It is worth noting that this behavior also persists in the black string case. Thus, we can state the following theorem:
\begin{theorem} \label{teorema2}
Static, cylindrically symmetric dyonic solutions to the Einstein-NED equations with an arbitrary $\mathcal{L}(\mathcal{
F})$ cannot describe spacetimes with a regular axis at $r=0$.
\end{theorem}

Another crucial point is that, although our metric ansatz is defined by Eq.~\eqref{metric}, the cosmological constant required for black string solutions does not affect the regularity conditions. Within the framework of NED, the criteria for a regular origin (or axis) depend solely on the behavior of the energy-momentum tensor components as $r \to 0$. Consequently, these conditions remain invariant whether the spacetime exhibits spherical or cylindrical symmetry.

\section{Charged singular black strings solutions}\label{sec:charged_bs}
As a first application, we construct black string solutions sourced by nonlinear electromagnetic fields that exhibit singular geometries, i.e., geometries that do not satisfy the regularity theorems discussed above. As illustrative examples, we consider two well-known NED models: the Euler--Heisenberg NED and logarithmic NED.
\subsection{Case I: Euler--Heisenberg Electrodynamics}
Initially derived from Dirac's electron--positron theory \cite{EulerHe}, the Euler--Heisenberg framework accounts for quantum vacuum effects, in which strong electromagnetic fields induce virtual electron--positron pair fluctuations. These vacuum polarization effects modify the classical dynamics of the electromagnetic field, leading to nonlinear corrections to Maxwell's equations. Following the work of Euler and Heisenberg, Schwinger derived the effective Lagrangian of quantum electrodynamics using the proper-time formalism. He showed that the renormalization of the field strength and electric charge, applied to the modified Lagrangian, leads to a finite and gauge-invariant result \cite{Schwinger1951}. This formulation reveals the emergence of nonlinear properties of the electromagnetic field in the quantum vacuum. The effective Lagrangian density is given by
\begin{equation} \label{EH}
\mathcal{L}(\mathcal{F})= -\mathcal{F} + \frac{\alpha}{2}\mathcal{F}^2,
\end{equation}
where the parameter $\alpha$ is related to the fine-structure constant and the electron mass. In the limit $\alpha \to 0$, Maxwell electrodynamics is recovered.

In order to obtain purely electric solutions, we employ the ansatz discussed in Eqs.~\eqref{eletrico}-\eqref{eletric}. With this choice, the equation for the electric field reduces to
\begin{equation}
    E(r) + 2\alpha E^3(r) = -\frac{q}{r^2}.
\end{equation}
This is a cubic equation with one real root and two complex conjugate roots. Due to the nonlinear structure of the electric field in this theory, it is not possible to obtain analytic solutions for $m(r)$ (Eq.\eqref{eletric}). Therefore, in what follows, we shall consider a purely magnetic configuration, for which analytic treatment becomes feasible.

For this purpose, we consider the ansatz given by Eqs.~\eqref{magnetico}-\eqref{magneticoeqfinal}. Substituting \eqref{EH} into \eqref{mu}, we obtain the solution
\begin{equation} \label{EHmetric}
   m (r)= \frac{\alpha  P^4 \ell ^3}{20 r^5}-\frac{P^2 \ell }{4 r} + \mu \implies f(r)=-\frac{4 \mu \ell}{r} +\frac{r^2}{\ell^2} + \frac{P^2 \ell^2 }{ r^2}-\frac{\alpha  P^4 \ell ^4}{ r^6}.
\end{equation}
For $r \rightarrow \infty$, we recover the usual vacuum black string solution. And for $r \rightarrow 0$, the $r^{-5}$ term dominates, indicating a strong singularity at $r=0$.

Analyzing the Kretschmann scalar,
\begin{align}
\mathcal{K}&=\frac{8}{{25 r^{16} \ell ^4}} \Big[239 \alpha ^2 P^8 \ell ^{14}- 380 \alpha  P^6 r^4 \ell ^{11}+25 P^4 r^8 \ell ^7 (7 \ell -2 \alpha ) +560  \alpha  \mu  P^4 r^5 \ell ^{10} \nonumber \\
&-600 \mu  P^2 r^9 \ell ^7 
+75 r^{16}+600 \mu ^2 r^{10} \ell ^6\Big]
\end{align}
We see that the NED strengthens the singularity, which now behaves as $r^{-16}$. As $r \rightarrow \infty$, $\mathcal{K}$ approaches a constant value, $24/\ell^4$.

Now, the curvature scalar $\mathcal{R}$,
\begin{equation}
  \mathcal{R} = \frac{4 \alpha P^4 \ell^5}{r^8} - \frac{12}{\ell^2},
\end{equation}
we observe that the presence of the NED introduces a singularity behaving as $r^{-8}$, unlike in the Maxwell case, where $\mathcal{R}$ is unaffected by the presence of charges. 

\subsection{Case II: Logarithmic electrodynamics}
In this subsection, we analyze the logarithmic model of NED, which has attracted considerable attention as a simple and physically well-motivated extension of Maxwell's theory. NED models are often introduced to incorporate corrections that become relevant in the strong-field regime and may regularize some divergences present in classical electrodynamics. In particular, logarithmic electrodynamics has been shown to possess improved behavior at large field strengths and to yield finite-energy configurations for point-like charges, while recovering Maxwell electrodynamics in the weak-field limit. Furthermore, this model has been widely employed in gravitational contexts, including black objects solutions sourced by NED fields. 

The corresponding Lagrangian density is given by \cite{Kruglov_2019, Soleng:1995kn, Hendi2012}
\begin{equation} \label{log}
    \mathcal{L}(\mathcal{F}) = -\beta^2 \ln\left(1 + \frac{\mathcal{F}}{\beta^2}\right),
\end{equation}
where $\beta$ is a parameter that controls the strength of the nonlinear corrections. In the weak-field limit $\mathcal{F} \ll \beta^2$, the Lagrangian reduces to the Maxwell form $\mathcal{L} \approx -\mathcal{F}$, ensuring the classical limit. In the limit where $\mathcal{F} \ll \beta^2$, $\mathcal{L}(\mathcal{F})$ can be expanded as
\begin{equation}
    \mathcal{L}(\mathcal{F}) \approx -\mathcal{F} + \frac{\mathcal{F}^2}{2 \beta^2} + \mathcal{O}(\mathcal{F}^3),
\end{equation}
and is thus equivalent to the Euler-Heisenberg Lagrangian for weak fields. Using \eqref{log} in \eqref{mu}, we obtain the purely magnetic solution:
\begin{align}
    m (r)&= -\frac{P^2 \ell}{{3 r}}  \, _2F_1\left(\frac{1}{4},1;\frac{5}{4};-\frac{2 P^2 \ell ^2}{r^4 \beta ^2}\right)+\frac{\beta ^2 r^3}{24 \ell }\ln \left(\frac{2 P^2 \ell ^2}{\beta ^2 r^4}+1\right)+\mu \implies  \\ \nonumber f(r)&=-\frac{4 \mu \ell}{r} + \frac{P^2 \ell^2}{{3 r^2}}  \, _2F_1\left(\frac{1}{4},1;\frac{5}{4};-\frac{2 P^2 \ell ^2}{r^4 \beta ^2}\right) -\frac{\beta ^2 r^2}{24 }\ln \left(\frac{2 P^2 \ell ^2}{\beta ^2 r^4}+1\right) ,
\end{align}
 where $_2F_1(a,b;c;x)$ is the Gauss' hypergeometric function. Expanding around $\mathcal{F}=0$, we obtain the following expression
\begin{equation}
m(r) \approx -\frac{P^2 \ell}{4r}+\frac{P^4 \ell^3}{20\beta^2 r^5}-\frac{P^6 \ell^5}{27\beta^4 r^9}+ \mathcal{O}(r^{-13}) +\mu .
\end{equation}

After an appropriate redefinition of the constants, this expression coincides with the solution obtained for the Euler--Heisenberg Lagrangian [see Eq.~\eqref{EHmetric}]. We also observe that, in the limit $\beta^2 \to \infty$, the nonlinear corrections vanish and the linear solution is recovered. In our previous analyses, we have shown that, in order for a purely magnetic solution to be regular at the origin, the Lagrangian density $\mathcal{L}(\mathcal{F})$ must approach a constant value in the limit $r \to 0$ (i.e., $\mathcal{F} \to \infty$). 
Notice that the present type of Lagrangian does not generate regular solutions, since $\mathcal{L}$ diverges in the limit $\mathcal{F} \to \infty$. 
Therefore, the nonlinear electrodynamics considered here is not capable of regularizing the geometry at the core.

Computing the Kretschmann scalar, we obtain
\begin{align}
    \mathcal{K}&=\frac{4 P^2 \ell ^3 \, _2F_1\left(\frac{1}{4},1;\frac{5}{4};-\frac{2 P^2 \ell ^2}{r^4 \beta ^2}\right) \left[P^2 \ell  \left(\beta ^2 r^3-48 \mu  \ell \right)-24 \beta ^2 \mu  r^4\right]}{3r^7 \left(2 P^2 \ell ^2+\beta ^2 r^4\right)}+\frac{4 P^4 \ell ^4 \, _2F_1\left(\frac{1}{4},1;\frac{5}{4};-\frac{2 P^2 \ell ^2}{r^4 \beta ^2}\right)^{2}}{3r^8} \nonumber\\
   & +\frac{1}{24}\beta ^4 \ln^{2}\left(\frac{2 P^2 \ell ^2}{\beta ^2 r^4}+1\right)-\frac{ \beta ^2 \left[P^2 \ell ^2 \left(\beta ^2 \ell ^2+24\right)+12 \beta ^2 r^4\right] \ln \left(\frac{2 P^2 \ell ^2}{\beta ^2 r^4}+1\right)}{6(2 P^2 \ell ^4+\beta ^2 r^4 \ell ^2)} \nonumber \\
   &+\frac{1}{r^6 \ell ^4 \left(2 P^2 \ell ^2+\beta ^2 r^4\right)^2} \Big\lbrace P^4 \ell ^4 \left[r^6 \left(\beta ^4 \ell ^4+8 \beta ^2 \ell ^2+96\right)- 32 \beta ^2 \mu  r^3 \ell ^5+768 \mu ^2 \ell ^6\right] \nonumber  \\
   &+ 4 \beta ^2 P^2 r^4 \ell ^2 \left[r^6 \left(\beta ^2 \ell ^2+24\right)-4 \beta ^2 \mu  r^3 \ell ^5+192 \mu ^2 \ell ^6\right]+24 \beta ^4 r^8 \left(r^6+8 \mu ^2 \ell ^6\right) \Big\rbrace
\end{align}

In the limit $r \to \infty$, $\mathcal{K}$ approaches a constant value, as in the vacuum solution. 
However, in this case, the nonlinear electrodynamics strengthens the singularity at the origin. 
Moreover, in the limit $\beta \to \infty$, the standard Maxwell solution is recovered.

Since $\mathcal{K}$ is singular as $r \to 0$, one expects that the Ricci scalar $\mathcal{R}$ exhibits the same behavior. Indeed, it is given by
\begin{equation}
  \mathcal{R}=  \frac{1}{2} \beta ^2 \ln \left(\frac{2 P^2 \ell ^2}{\beta ^2 r^4}+1\right)-\frac{P^2 \ell ^2 \left(\beta ^2 \ell ^2+24\right)+12 \beta ^2 r^4}{2 P^2 \ell ^4+\beta ^2 r^4 \ell ^2}.
\end{equation}

Finally, similarly to the Euler--Heisenberg case, purely electric solutions cannot be obtained in analytic form.

\section{Charged Regular Black strings solutions}\label{sec:regular_bs}
As a second application, we construct regular black string solutions supported by nonlinear electromagnetic fields described by the Bardeen and Hayward NED models. In contrast to the previous examples, these theories are known to generate regular geometries, in the sense that the spacetime curvature invariants remain finite throughout the manifold. In particular, the corresponding nonlinear electromagnetic sources satisfy the conditions required for the existence of regular solutions discussed in the previous sections. These models have been widely studied in the context of regular black holes, where the central singularity is replaced by a regular core sustained by nonlinear electromagnetic effects. Here, we investigate their cylindrical counterparts and analyze the extent to which these regularity properties are preserved in the black string configuration.

\subsection{Bardeen Class}
In this subsection, we obtain a regular charged black string solution by considering the nonlinear electrodynamics (NED) Lagrangian density originally employed by Ayon-Beato and Garcia \cite{Ay_n_Beato_2000}. This specific Lagrangian is well-known in the literature for acting as the physical matter source that generates the Bardeen regular black hole in spherical symmetry \cite{Bardeen1968}. This Lagrangian is given by:
\begin{equation}
    \mathcal{L}(\mathcal{F}) = -\frac{3}{g P^2} 
    \left( 
        \frac{\sqrt{P^2\ell^2\mathcal{F}/2}}
        {1+\sqrt{P^2\ell^2\mathcal{F}/2}}
    \right)^{5/2},
\end{equation}
where $g = |P \ell^2| / 8\mu$. Writing $\mathcal{L}(\mathcal{F})$ in terms of $r$, we obtain
\begin{equation}
  \mathcal{L}(r) = -\frac{3 P^3 \ell ^5}{g \left(P^2 \ell ^2+r^2\right)^{5/2}}.
\end{equation}
In this way, we see that $\mathcal{L}(0) = -{3}/{g P^3 \ell ^2} = \mathcal{L}_0$
is regular at the origin and assumes a constant value, indicating that this type of Lagrangian produces regular solutions, as discussed previously. And for limit $r\to \infty\ (\mathcal{F}\to0)$ we have $\mathcal{L}(\mathcal{F}(r))\to0$.

Since this is a purely magnetic Lagrangian, it suffices to solve Eq.~\eqref{mu}:
\begin{equation}
    \frac{3 P^3 \ell ^5}{2 g \left(P^2 \ell ^2+r^2\right)^{5/2}}
    -\frac{4 m '(r) \ell}{r^2}=0,
\end{equation}
whose solution is
\begin{equation}
    m (r)=\frac{\mu r^3}{\left(P^2 \ell ^2+r^2\right)^{3/2}} \implies f(r)=\frac{-4\mu \ell r^2}{\left(P^2 \ell ^2+r^2\right)^{3/2}} +\frac{r^2}{\ell^2}
\end{equation}
so that, as $r \to 0$, we have $m(r) \approx \mathcal{O}(r^3)$, confirming the regularity condition proposed in \eqref{regularity1}.The event horizon is located at $r_h = \ell \sqrt{(4\mu)^{2/3} - P^2}$. For $P^2 = (4\mu)^{2/3}$, we obtain an extremal horizon ($r_h = 0$), whereas for $P^2 > (4\mu)^{2/3}$, no horizons exist. 

Expanding the metric for $r \to 0$, we have:
\begin{equation}
  f(r) \approx r^2 \left( \frac{1}{\ell^2} - \frac{4\mu \ell}{(P^2 \ell^2)^{3/2}} \right) + \mathcal{O}(r^3),
\end{equation}
where for $P > (4\mu)^{1/3}$ we find an Anti-de Sitter core, and for $P < (4\mu)^{1/3}$ we obtain a de Sitter core. In the extremal case, $P = (4\mu)^{1/3}$. Notably, the threshold for the existence of the event horizon coincides with the geometric transition of the core. In the limit $r\to \infty$ we recover the vacuum solution.

Calculating the Kretschmann scalar, we obtain the following expression
\begin{equation}
\mathcal{K} = 24 \Bigg[
\frac{2 \mu^2 \ell^2 \left(8 P^8 \ell^8 -4 P^6 r^2 \ell^6 +47 P^4 r^4 \ell^4 -12 P^2 r^6 \ell^2 +4 r^8\right)}
{\left(P^2 \ell^2 + r^2\right)^7} 
+ \frac{1}{\ell^4}+ 2 \mu P^2 \ell (r^2 - 4 P^2 \ell^2)
\left(\frac{1}{P^2 \ell^2 + r^2}\right)^{7/2}
\Bigg]
\end{equation}
From this expression, we see that in the limit $r \to 0$, $\mathcal{K}$ remains finite, ensuring the regularity of the solution at the origin. In the asymptotic limit $r \to \infty$, we recover the vacuum solution, for which $\mathcal{K}$ becomes a constant.

Now, calculating the Ricci scalar, we obtain
\begin{equation}
    \mathcal{R}=-\frac{12}{\ell^2} \left[ \mu P^2 \ell ^5 \left(r^2-4 P^2 \ell ^2\right) \left(\frac{1}{P^2 \ell ^2+r^2}\right)^{7/2}+1\right].
\end{equation}
Since the Kretschmann scalar is regular at the origin, it follows that all other curvature invariants are also regular there. In the asymptotic limit $r \to \infty$, $\mathcal{R} \to -12/\ell^2$.

In addition, we can obtain an electric-type source for a Bardeen-like black string. A source of this kind is already known in the spherically symmetric case (black hole) \cite{Rodrigues:2018bdc}. Here, however, we are interested in determining the corresponding source in cylindrical symmetry (black string).

The general expression for the electric field is given in Eq.~\eqref{campoeletrico}. Rather than fixing the NED model beforehand, we reconstruct the theory through a reverse-engineering procedure. Specifically, we start from the mass function $m(r)$ obtained from the purely magnetic solution and substitute it into Eqs.~\eqref{eletric1} and \eqref{eletric}. By solving these equations, we determine the functions $\mathcal{L}(r)$ and $\mathcal{L}_F(r)$ that generate the same geometry in the purely electric case. In this way, we obtain
\begin{equation}
    \mathcal{L}(r)=12 \mu \ell ^3 \left(\frac{1}{q^2 \ell ^2+r^2}\right)^{7/2} \left(3 q^2 r^2-2 q^4 \ell ^2\right),
\end{equation}
\begin{equation}
 \mathcal{L}_F(\mathcal{F}(r))=   -\frac{(q^2 \ell ^2+r^2)^{7/2}}{15 \mu r^6 \ell ^3 } \quad \text{and} \quad \mathcal{F}(r)=-\frac{450 \mu ^2 q^2 r^8 \ell ^6}{\left(q^2 \ell ^2+r^2\right)^7}.
\end{equation}

The above expressions satisfy the consistency relation $\mathcal{L}_F\,\mathcal{F}'(r)-\mathcal{L}'(r)=0$, ensuring that $\mathcal{L}_F=d\mathcal{L}/d\mathcal{F}$. Therefore, $\mathcal{L}(r)$ and $\mathcal{L}_F(r)$ consistently define a unique nonlinear electrodynamics model in parametric form. However, due to the highly nontrivial dependence of $\mathcal{F}(r)$ on the radial coordinate, it is not possible to invert $\mathcal{F}(r)$ analytically to obtain $r=r(\mathcal{F})$. Consequently, the Lagrangian cannot be expressed explicitly as a closed-form function $\mathcal{L}(\mathcal{F})$, and the theory is instead characterized by the parametric representation $\{\mathcal{L}(r),\mathcal{F}(r)\}$.

Returning to Eq.~\eqref{campoeletrico}, we obtain the electric field associated with this Lagrangian,
\begin{equation}
    E(r)=-15 \mu q r^4 \ell ^3 \left(\frac{1}{q^2 \ell ^2+r^2}\right)^{7/2}.
\end{equation}

Analyzing the three previous equations, we observe that this type of Lagrangian does not recover Maxwell electrodynamics in the limit $r \to \infty$, confirming the conditions established earlier: if we require a purely electric solution that is regular at the origin, then it cannot recover Maxwell behavior in the asymptotic limit $r \to \infty$. In fact, we have
\begin{equation}
    E(r) \to \frac{1}{r^3} \quad \text{for} \quad r \to \infty,
\end{equation}
confirming that $E(r)$ does not recover the Coulomb behavior in the weak-field regime.

\subsection{Hayward class}

The Hayward solution describes a regular black hole in the spherical case, such that as $r \to \infty$ the Schwarzschild metric is recovered, and as $r \to 0$ the spacetime forms a de Sitter-like core \cite{Hayward:2005gi}. Hayward, however, did not obtain the matter or field source that supports this type of spacetime. Thus, Fan and Wang \cite{Fan:2016hvf}, by means of a general model of nonlinear electrodynamics (NED) using magnetic monopoles, obtained the source Lagrangian for this solution. 

Based on this, the Lagrangian density is given by
\begin{equation} \label{hayward}
    \mathcal{L}(\mathcal{F}) = -\frac{3 (\alpha \mathcal{F})^{3/2}}{\alpha \left[1+(\alpha \mathcal{F})^{3/4}\right]^2}
\end{equation}
which, in terms of the radial coordinate, becomes:
\begin{equation}
    \mathcal{L}(\mathcal{F}(r))=-\frac{3(2\alpha \ell^2 P^2)^{3/2}}{\alpha \left[r^3+(2 \alpha \ell^2 P^2)^{3/4}\right]^2}.
\end{equation}
Again, we observe that as $r \to 0$, $\mathcal{L}(r)$ tends toward a constant value, namely $\mathcal{L} = -3/\alpha$. Since this Lagrangian generates regular solutions, this behavior was already expected, as shown in Section \ref{sec:main_equations}. Thus, we can solve Eq. \eqref{mu}, resulting in 
\begin{equation}
    m(r)=\frac{\mu r^3}{(r^3 + \lambda^3)} \implies f(r)=-\frac{4 \mu\ell r^2}{(r^3+\lambda^3)} +\frac{r^2}{\ell^2},
\end{equation}
where we have defined $\lambda^3=(2 \alpha \ell^2 P^2)^{3/4}$ and $\alpha=\lambda^3/8\ell \mu$, with the integration constant $C=\mu$. For $r \to \infty$, we obtain $f(r) \approx -4\mu\ell/r + r^2/\ell^2$ (the vacuum solution). The event horizon is located at $r_h = (4 \mu \ell^3 - \lambda^3)^{1/3}$.

Expanding the solution for $r \to 0$:
\begin{equation}
f(r) \approx r^2 \left(\frac{1}{\ell^2} - \frac{4 \mu \ell}{\lambda^3}\right) + \frac{4 \mu r^5 \ell}{\lambda^6} + \mathcal{O}\left(r^8\right).
\end{equation}
The term proportional to $r^2$ dictates the effective geometry of the core, which is determined by the value of $\lambda$. For $\lambda > (4\mu)^{1/3}\ell$, the core exhibits an Anti-de Sitter (AdS) behavior, whereas for $\lambda < (4\mu)^{1/3}\ell$, a de Sitter (dS) core is obtained. At the critical point $\lambda^3 = 4\mu\ell^3$, the event horizon coincides with the origin, and the spacetime becomes locally flat at the center as the leading-order curvature terms vanish.

Now, computing the Kretschmann scalar, we obtain
\begin{equation}
\begin{split}
\mathcal{K}=&\frac{24}{\ell ^4 \left(\lambda ^2+r^3\right)^6} 
\Big[r^{18} +6 \lambda ^2 r^{15}
+r^{12} \left(15 \lambda ^4+8 \mu ^2 \ell ^6+4 \lambda ^2 \mu  \ell ^3\right) 
+4 r^9 \left(5 \lambda ^6-8 \lambda ^2 \mu ^2 \ell ^6+\lambda ^4 \mu  \ell ^3\right) \\
&+3 r^6 \left(5 \lambda ^8+48 \lambda ^4 \mu ^2 \ell ^6-4 \lambda ^6 \mu  \ell ^3\right)
+2 \lambda ^6 r^3 \left(3 \lambda ^4-8 \mu ^2 \ell ^6-10 \lambda ^2 \mu  \ell ^3\right)
+\lambda ^8 \left(\lambda ^2-4 \mu  \ell ^3\right)^2\Big]
\end{split}
\end{equation}
which is regular at the origin, and in the limit $r\to \infty$ approaches a constant value, as expected for a black string solution. As with the Kretschmann scalar, the Ricci scalar must also be regular:
\begin{equation}
 \mathcal{R}=   -\frac{24 \lambda ^2 \mu  \ell  \left(r^3-2 \lambda ^2\right)}{\left(\lambda ^2+r^3\right)^3}-\frac{12}{\ell ^2}
\end{equation}
and it also recovers the same asymptotic value in the limit $r\to \infty$.

Finally, we can obtain a purely electric source for the Hayward class. To do so, we follow the same procedure used for the Bardeen-type solution. Starting from the mass function $m(r)$ obtained from a magnetically charged source, we solve Eqs.~\eqref{eletric1} and \eqref{eletric} to determine $\mathcal{L}(\mathcal{F})$ and $\mathcal{L_F}$. In this way, we obtain
\begin{equation}
    \mathcal{L}(\mathcal{F}(r))=\frac{24 \mu  \ell  \left(2 \lambda ^2 r^3-\lambda ^4\right)}{\left(\lambda ^2+r^3\right)^3}, \quad    
    \mathcal{L_F}(\mathcal{F}(r))=-\frac{q^2 \left(\lambda ^2+r^3\right)^3}{18 \lambda ^2 \mu  r^7 \ell }.
\end{equation}

Since this is a regular Lagrangian, it is expected that the Maxwell limit is not recovered for $r\to\infty$, where $\mathcal{F}\approx0$. To verify this, let us examine the electric field of this NED model:
\begin{equation}
    E(r)=-\frac{18 \lambda ^2 \mu  r^5 \ell }{q \left(\lambda ^2+r^3\right)^3}.
\end{equation}

In the asymptotic limit $r\to\infty$, we obtain
\begin{equation}
    E(r)\to \frac{1}{r^4},
\end{equation}
confirming that the electric field does not behave as in Maxwell theory in the weak-field regime.
\section{Causality and unitarity}\label{sec:causality}
In classical Maxwell theory, photons propagate along null geodesics of the background metric $g_{\mu\nu}$, thereby defining the boundary of the light cone. To preserve causality, no signal can propagate outside this cone. In nonlinear electrodynamics, however, the interactions become strong enough that the electromagnetic field self-interacts, leading to vacuum polarization effects. Small perturbations propagating over a nontrivial background field no longer follow null geodesics of $g_{\mu\nu}$ but instead propagate according to an effective metric $g^{\mathrm{eff}}_{\mu\nu}$. In the geometrical optics (eikonal) approximation, the wave vectors of these perturbations satisfy null conditions with respect to $g^{\mathrm{eff}}_{\mu\nu}$, determining the modified causal structure induced by nonlinearities. 

The equations that guarantee causality and unitarity, in our notation, are as follows:
\begin{equation}\label{causalidade}
\mathcal{L_F} < 0 \ ,\ \mathcal{L_{FF}} \geq 0 \quad \text{and} \quad \Phi = \mathcal{L_F} + 2\mathcal{F}\mathcal{L_{FF}} \leq 0.
\end{equation}
From these equations, we can state the following theorem regarding purely magnetic solutions \cite{Bronnikov2023}:
\begin{theorem} \label{teoremac}
    Regular solutions with cylindrical symmetry generated by purely magnetic NEDs necessarily violate causality in regions near the axis $r=0$.
\end{theorem}

\begin{proof}
We can write $\Phi = 2 \sqrt{\mathcal{F}}\frac{d(\mathcal{L_F} \sqrt{\mathcal{F}})}{d\mathcal{F}}$. Therefore, a regular NED at the origin requires $\mathcal{L}(\mathcal{F})$ to assume a finite value in the limit $\mathcal{F} \to \infty$. For the integral $\mathcal{L} = \int \mathcal{L_F} d\mathcal{F}$ to be finite, $\mathcal{L_F}$ must tend to zero sufficiently fast for large $\mathcal{F}$. Consequently, we would have $\Phi > 0$. Since our initial premise was that $\mathcal{L_F} < 0$, it is impossible to satisfy the first and third conditions simultaneously.
\end{proof}

Our first two cases correspond to purely magnetic solutions that are singular at the origin. According to the previous theorem, one should expect the conditions \eqref{causalidade} to be satisfied.

For the Euler--Heisenberg Lagrangian we have
\begin{equation}
    \mathcal{L}_F = 2\alpha \mathcal{F} - 1, \quad 
    \mathcal{L}_{FF} = 4\alpha > 0 
    \quad \text{and} \quad 
    \Phi = 6\alpha \mathcal{F} - 1 .
\end{equation}

In this case, the correction term $\alpha \mathcal{F}^2$ is only valid for sufficiently small fields. Therefore, causality is always preserved.

For the logarithmic Lagrangian, we obtain
\begin{equation}
    \mathcal{L_F} = -\frac{\beta^2}{\beta^2+\mathcal{F}} < 0, \quad \mathcal{L_{FF}} = \frac{\beta^2}{\left(\beta^2+\mathcal{F}\right)^2} > 0 \quad \text{and} \quad \Phi = \frac{\beta^2 \left(\mathcal{F}-\beta^2\right)}{\left(\beta^2+\mathcal{F}\right)^2},
\end{equation}
where the first two conditions are globally satisfied, but the third causality condition ($\Phi \leq 0$) strictly requires $\mathcal{F} \leq \beta^2$. Expanding the parameter $\Phi$ near the origin, we obtain:
\begin{equation}
    \Phi \approx \frac{\beta^2 r^4}{2 P^2 \ell^2} - \frac{3 \beta^4 r^8}{4 P^4 \ell^4} + \mathcal{O}\left(r^9\right).
\end{equation}
Although we have $r\to0$ ,$\Phi = 0$, the expansion demonstrates that $\Phi$ approaches zero taking strictly positive values. Since the solution is singular, there will exist a critical value of $r$ at which $\mathcal{F}$ exceeds the value $\beta^2$.  Consequently, causality is violated in a neighborhood of the center. This result illustrates that causality violation is not exclusive to regular black holes; models with a central singularity can also suffer from superluminal propagation and instabilities if the magnetic field exceeds the critical threshold established by the parameter $\beta$.

Now, for the Ayon--Beato and Garcia Lagrangian we obtain

\begin{equation}
    \mathcal{L}_F =
    -\frac{15 P r^6 \ell^3}{8 g (P^2 \ell^2 + r^2)^{7/2}} < 0, \quad
    \mathcal{L}_{FF} =
    -\frac{15 r^{10} \ell (r^2 - 6P^2 \ell^2)}
    {64 P g (P^2 \ell^2 + r^2)^{9/2}}
\end{equation}

\begin{equation}
    \Phi =
    -\frac{15 P r^6 \ell^3 (3r^2 - 4P^2 \ell^2)}
    {16 g (P^2 \ell^2 + r^2)^{9/2}} .
\end{equation}

Therefore, the condition on $\Phi$ is satisfied only for $r \geq 2P\ell/\sqrt{3}$. In the region $0 < r < 2P\ell/\sqrt{3}$ this condition is violated and the spacetime ceases to be causal, as predicted by the previous theorem, which was already known to hold for spherical symmetry and whose validity we now confirm also for cylindrical symmetry. This result indicates that, for this Lagrangian, the electromagnetic field becomes superluminal, meaning that photons propagate outside the light cone. In this region the theory loses its predictive power regarding the future evolution of the system.

Finally, for the Hayward case, we have
\begin{equation}
    \mathcal{L}_F = -\frac{9 \sqrt{\alpha \mathcal{F}}}{2 \left[(\alpha \mathcal{F})^{3/4}+1\right]^3} < 0, \quad \mathcal{L}_{FF} = \frac{9 \alpha \left[7 (\alpha \mathcal{F})^{3/4}-2\right]}{8 \sqrt{\alpha \mathcal{F}} \left[(\alpha \mathcal{F})^{3/4}+1\right]^4} \,,
\end{equation}
\begin{equation}
   \Phi = \frac{9 \sqrt{\alpha \mathcal{F}} \left[5 (\alpha \mathcal{F})^{3/4}-4\right]}{4 \left[(\alpha \mathcal{F})^{3/4}+1\right]^4} \,.
\end{equation}
We can observe that the condition on $\mathcal{L}_{FF}$ requires $\alpha \mathcal{F} \geq (2/7)^{4/3} \approx 0.188$. Since $\mathcal{F} \propto r^{-4}$, this condition is violated as $r \to \infty$, because in this limit $\mathcal{F}$ assumes very small values. This result occurs due to the behavior of this NED model in the weak-field regime, where it does not recover Maxwell electrodynamics. On the other hand, the condition on $\Phi$ imposes that $\alpha \mathcal{F} \leq (4/5)^{4/3} \approx 0.743$. For large values of $r$, this condition is satisfied; however, near the axis (origin), $\mathcal{F}$ diverges, exceeding the permitted value and thus violating the condition. Therefore, similar to the Bardeen case, we have a region where causality is broken.

\section{Conclusions}\label{sec:conclusions}
In this work, we have systematically investigated the coupling of General Relativity with NED within four-dimensional cylindrical topologies. Our primary objective was to resolve the intrinsic axial curvature singularity of classical black strings through the gravitational backreaction of nonlinear electromagnetic fields governed by a general Lagrangian $\mathcal{L}(\mathcal{F})$.

Through a model-independent theoretical framework, we extended established regularity conditions from spherically symmetric spacetimes to cylindrical geometries. A central result of our analysis is the explicit demonstration that well-known no-go theorems strictly hold for black strings. Specifically, we proved that regular, purely electric configurations, as well as dyonic ones, cannot be generated by any NED Lagrangian that recovers the standard Maxwell limit, $\mathcal{L}(\mathcal{F}) \approx -\mathcal{F}$, in the weak-field regime establishing a fundamental constraint for axial symmetry. Our results are consistent with those reported in Refs.\cite{Bronnikov:2003ru,Bronnikov:2000vy,Bronnikov2023}, where both cylindrically and spherically symmetric configurations were investigated. The main distinction of the present work is that we consider the Lemos black-string geometry, a cylindrically symmetric black-hole solution with Anti-de Sitter asymptotics. In contrast to the horizonless solitonic configurations studied in Ref.\cite{Bronnikov:2003ru}, the solutions analyzed here may possess event horizons and are supported by a negative cosmological constant.

Applying these foundational theorems, we constructed and analyzed specific exact solutions. First, we derived singular black string metrics sourced by Euler--Heisenberg and logarithmic electrodynamics, demonstrating how nonlinearities modify the spacetime structure while still harboring an axial singularity, consistent with the aforementioned no-go constraints. Subsequently, moving to models that circumvent these limitations, we successfully generated entirely regular black string solutions belonging to the Bardeen and Hayward classes. In these purely magnetic and electric configurations, the severe axial singularity is effectively smoothed out by the nonlinear electromagnetic environment, confirming our theoretical predictions.

Finally, because singularity resolution via NED often introduces potential physical pathologies, we subjected our solutions to stringent causality and unitarity constraints. Contrary to the ideal scenario where superluminal propagation is entirely avoided, our analysis highlights a fundamental physical tension. We explicitly demonstrated, through a general theorem, that regular purely magnetic black strings necessarily violate causality in the deep core region. This was verified in our exact solutions: while the weak-field Euler--Heisenberg model preserves causal structure, both the singular logarithmic model and the regular configurations exhibit superluminal photon propagation near the axis. Consequently, there exists a critical radius within which the effective geometry permits signals to propagate outside the light cone, signaling a breakdown of classical predictability in the innermost region.

Despite these causal limitations near the core, the successful generalization of regular NED black holes to cylindrical topologies, alongside the rigorous mapping of their physical constraints, opens multiple board avenues for future research. Since the exterior geometries remain well-behaved, immediate investigations could focus on their thermodynamic properties, including the analysis of phase transitions, and their classical stability against scalar and tensor perturbations through the study of quasinormal modes. Furthermore, the robust exterior structure invites detailed studies of particle motion, geodesic trajectories, and potentially the optical appearance of these extended objects.

Extending our static framework to include angular momentum, thereby constructing rotating regular black strings, represents another highly relevant continuation. From a higher-dimensional perspective, generalizing these NED couplings to black branes could yield deeper insights into the AdS/CFT correspondence and holographic fluid dynamics. Finally, exploring quantum effects in these backgrounds, such as Hawking radiation and vacuum polarization, could further illuminate the intricate interplay between quantum field theory and regularized spacetime geometries.

\acknowledgments{The authors G.A. and M.S.C. would like to thank Conselho Nacional de Desenvolvimento Cient\'{i}fico e Tecnol\'{o}gico (CNPq) and Fundação Cearense de Apoio ao Desenvolvimento Científico e Tecnológico (FUNCAP) through PRONEM PNE0112- 00085.01.00/16, for the partial financial support. V.H.U.B. and T.M.C is supported by Coordena\c c\~{a}o de Aperfei\c coamento de Pessoal de N\'{i}vel Superior - Brasil (CAPES) - Finance Code 001.
}
\bibliographystyle{apsrev4-2}
\bibliography{ref}
\end{document}